\documentclass[conf, letter, 10pt]{IEEEtran}
\usepackage[utf8]{inputenc}
\usepackage{amsmath}
\usepackage{amsfonts}
\usepackage{amssymb}
\usepackage{amsthm}
\usepackage{commath}
\usepackage{graphicx}
\usepackage{bm}
\usepackage{multirow}
\usepackage{epstopdf}
\usepackage{physics}
\usepackage{mathtools}
\usepackage{cases}
\usepackage{enumitem}
\usepackage{cite}
\usepackage{xcolor}

\usepackage{caption}
\allowdisplaybreaks
\newtheorem{theorem}{Theorem}
\newtheorem{definition}{Definition}

\theoremstyle{definition}
\newtheorem{remark}{Remark}
\theoremstyle{definition}
\newtheorem{example}{Example}

\makeatletter
\renewcommand*\env@matrix[1][\arraystretch]{%
  \edef\arraystretch{#1}%
  \hskip -\arraycolsep
  \let\@ifnextchar\new@ifnextchar
  \array{*\c@MaxMatrixCols c}}
\makeatother

\title{Lossy Transmission of Correlated Sources over Two-Way Channels}
\renewcommand{\arraystretch}{1.5}

\newcommand{\markov}{\mathrel{\multimap}\joinrel\mathrel{-}%
\joinrel\mathrel{\mkern-6mu}\joinrel\mathrel{-}}
\newcommand{\mli}[1]{\mathit{#1}}
\begin{document}

\title{Joint Source-Channel Coding for the Transmission of Correlated Sources over Two-Way Channels}

\author{
\IEEEauthorblockN{Jian-Jia Weng, Fady Alajaji, and Tam\'as Linder}
\thanks{%
The authors are with the Department of Mathematics and Statistics, Queen's University, Kingston, ON K7L 3N6, Canada (Emails: jian-jia.weng@queensu.ca, \{fady, linder\}@mast.queensu.ca).}
\thanks{
This work was supported in part by NSERC of Canada.}
}

\maketitle
\begin{abstract}
A joint source-channel coding (JSCC) scheme based on hybrid digital/analog coding is proposed for the transmission of correlated sources over discrete-memoryless two-way channels (DM-TWCs). 
The scheme utilizes the correlation between the sources in generating channel inputs, thus enabling the users to coordinate their transmission to combat channel noise. 
The hybrid scheme also subsumes prior coding methods such as rate-one separate source-channel coding and uncoded schemes for two-way lossy transmission, as well as the correlation-preserving coding scheme for (almost) lossless transmission. 
Moreover, we derive a distortion outer bound for the source-channel system using a genie-aided argument. 
A complete JSSC theorem for a class of correlated sources and DM-TWCs whose capacity region cannot be enlarged via interactive adaptive coding is also established. 
Examples that illustrate the theorem are given. 
\end{abstract}


\section{Introduction}\label{sec:introduction}
The transmission of correlated sources over discrete-memoryless two-way channels (DM-TWCs) was first considered by Shannon in \cite[Sec. 14]{shannon1961}. 
Shannon exhibited an uncoded scheme which preserves source correlation and attains the mid-point of his capacity outer bound curve for binary multiplying channels.  
He observed that error-free transmission is feasible if the source and channel statistics are perfectly matched. 
Along with determining channel capacity, the problem of how to efficiently transmit correlated sources over DM-TWCs is still largely unsolved.  
In particular, the perfect matching of source-channel statistics is not feasible except in some very special cases.
Here, we investigate the transmission performance of the two-way source-channel system as shown in Fig.~\ref{fig:TWCblcok}, where two users exchange correlated information under fidelity constraints. 

In \cite{gunduz2009}, the authors adopted the correlation-preserving coding scheme of \cite{cover1980} for (almost) lossless transmission of correlated sources. 
Similar to Shannon's idea, the coding scheme in \cite{gunduz2009} aims to preserve source correlation in the channel inputs to facilitate two-way transmission. 
Nevertheless, this scheme does not apply to the lossy setup. 
In \cite{jjw2017}, a separate lossy source-channel coding (SSCC) scheme that decouples data compression and error correction was constructed. 
Here, the two users employ Wyner-Ziv (WZ) lossy coding \cite{wyner1976} in tandem with standard (non-adaptive) channel coding.   
Different from the coding scheme in \cite{gunduz2009}, the separated coding structure disables the use of source correlation in generating channel inputs. 
Also, the SSCC scheme is generally sub-optimal for transmitting correlated sources. It is thus of interest to develop other more efficient coding methods.  


\begin{figure}[!t]
\centering
\includegraphics[draft=false, scale=0.53]{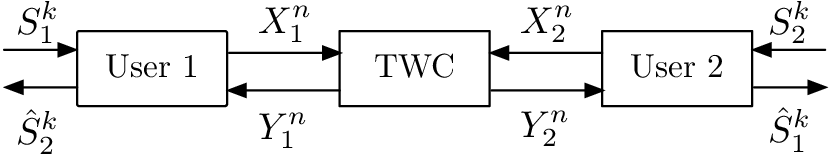}
\caption{The block diagram for the lossy transmission of correlated source $(S_1^k, S_2^k)$ via $n$ uses of a DM-TWC.}
\label{fig:TWCblcok}
\end{figure}

In this paper, we first establish a joint (lossy) source-channel achievability theorem via a hybrid coding scheme. 
Based on the viewpoint that a TWC consists of two state-dependent one-way channels, we extend the hybrid digital/analog coding scheme of \cite{kim2015} to TWCs so that the new scheme takes into account both channel state and receiver side-information.   
The proposed scheme and the achievability result recover prior results in \cite{gunduz2009} and \cite{jjw2017} as special cases.
{The proposed coding scheme is also illustrated via an example.}  
Moreover, we present a genie-aided distortion outer bound. 
Based on this outer bound and the SSCC scheme given in \cite{jjw2017}, we derive a joint source-channel coding theorem for a class of sources and DM-TWCs. 
We note that this paper concerns two-way \emph{simultaneous} transmission. 
Results for interactive two-way communication can be found in \cite{kaspi1985, maor2006}. 

The rest of this paper is organized as follows. 
In Section~II, the system model and definitions are introduced.  
The proposed joint source-channel hybrid coding scheme and special cases are presented in Section~III. 
The genie-aided outer bound and a complete joint source-channel coding theorem are given in Section~IV and conclusions are drawn in Section~V. 

\section{Preliminaries}\label{sec:p2p}
As shown in Fig.\ref{fig:TWCblcok},  two users simultaneously exchange correlated sources over a noisy DM-TWC subject to fidelity constraints.
For $j=1, 2$ and some $k\ge 1$, let $S_j^{k}\triangleq(S_{j, 1}, S_{j, 2}, \dots, \allowbreak S_{j, k})$ denote user-$j$'s source sequence with $S_{j, i}\in\mathcal{S}_j$, where $\mathcal{S}_j$ is the source alphabet. 
Also, let $\hat{S}_j^{k}\triangleq(\hat{S}_{j, 1}, \hat{S}_{j, 2}, \dots, \allowbreak \hat{S}_{j, k})$ denote the reconstruction of $S_j^{k}$ with $\hat{S}_{j, i}\in\hat{\mathcal{S}}_j$, where $\mathcal{\hat{S}}_j$ is the reconstruction alphabet. 
The source pair $\{(S_{1, m}, S_{2, m})\}_{m=1}^k$ is assumed to be memoryless in time having joint probability distribution $P_{S_1, S_2}$ 
at each time instant, i.e., $P_{S_1^k, S_2^k}(s_1^k, s_2^k)=\prod_{m=1}^k P_{S_1, S_2}(s_{1, m}, s_{2, m})$. 
The reconstruction quality for user~$j$ is assessed via the distortion measure $d_j: \mathcal{S}_j\times\mathcal{\hat{S}}_j\rightarrow\mathbb{R}_{\ge 0}$.
The distortion between sequences $s_j^k$ and $\hat{s}_j^k$ is defined as $d_j(s_j^k, \hat{s}_j^k)\triangleq k^{-1}\sum_{m=1}^k d_j({s_{j, m}, \hat{s}_{j, m}})$ for $j=1, 2$.

Let $\mathcal{X}_j$ and $\mathcal{Y}_j$ denote the channel input and output alphabets for user $j$ and let $X_j^{n}=(X_{j, 1}, X_{j, 2}, \dots, \allowbreak X_{j, n})$ and $Y_{j}^n=(Y_{j, 1}, Y_{j, 2}, \dots, \allowbreak Y_{j, n})$ denote the channel input and received sequences, respectively. 
Here, except when explicitly stated otherwise, 
we assume that all alphabets are finite.
We next define a joint source-channel code (JSCC) for transmitting $S_1^k$ and $S_2^k$ via $n$ channel uses of a DM-TWC with transition distribution $P_{Y_1,Y_2|X_1,X_2}$. 

\begin{definition}
An $(n, k)$ JSCC for the lossy transmission of $(S_1^k, S_2^k)$ over a DM-TWC consists of two sequences of encoding functions $f_1\triangleq \{f_{1, i}\}_{i=1}^n$ and $f_2\triangleq \{f_{2, i}\}_{i=1}^n$ such that
\[
\begin{array}{ll}
X_{1,1}= f_{1, 1}(S_1^k),& X_{1,i}=f_{1, i}(S_1^k, Y_1^{i-1})\\
X_{2,1}=f_{2, 1}(S_2^k),& X_{2,i}=f_{2, i}(S_2^k, Y_2^{i-1})
\end{array}
\]
for $i=2, 3, \dots, n$, and two decoding functions $g_1$ and $g_2$ such that $\hat{S}_2^k=g_1(S_1^k, Y_1^n)$ and $\hat{S}_1^k=g_2(S_2^k, Y_2^n)$.
\end{definition}

Note that with this code definition, the joint distribution of all the system variables is given by 
\begin{IEEEeqnarray}{l}
P_{S_1^k, S_2^k, X_1^n, X_2^n, Y_1^n, Y_2^n}=P_{S_1^k, S_2^k}\cdot \left(\prod\limits_{i=1}^n P_{X_{1, i}|S_1^k, Y_1^{i-1}}\right)\cdot\nonumber\\
\ \ \qquad\qquad\left(\prod\limits_{i=1}^n P_{X_{2, i}|S_2^k, Y_2^{i-1}}\right)\cdot\left(\prod\limits_{i=1}^n P_{Y_{1,i}, Y_{2, i}|X_{1, i}, X_{2, i}}\right),\nonumber
\end{IEEEeqnarray}
where $P_{Y_{1,i}Y_{2,i}|X_{1,i},X_{2,i}}=P_{Y_1,Y_2|X_1,X_2}$ for $i=1,\ldots,n$.
Also, the rate of the code is $k/n$ (in source symbol/channel use), and the expected distortion associated with the code is $D_j(k)=\mathbb{E}[d_j(S_j^k, \hat{S}_j^k)]$ for $j=1, 2$.

\begin{definition}
A distortion pair $(D_1, D_2)$ is said to be achievable at rate $R$ over a DM-TWC if there exists a sequence of $(n, k)$ JSCCs (where $n$ is a function of $k$) such that for $j=1, 2$, $\lim_{k\to\infty} k/n = R$ and $\limsup_{k\to\infty}D_j(k)\le D_j$. 
The achievable distortion region of a rate $R$ two-way lossy transmission system is defined as the convex closure of the set of all achievable distortion pairs. 
\end{definition}

We will need the following known definitions. Here, $j=1, 2$ with $j\neq j'$.
\begin{itemize}[leftmargin=*]
\item Standard RD function \cite[Sec. 3.6]{kim2011}: 
\[
R^{(j)}(D_j)=\min\limits_{P_{\hat{S}_j|S_j}:\mathbb{E}[d_j(S_j, \hat{S}_j)]\le D_j} I(S_j; \hat{S}_j).
\]
\smallskip
\item Wyner-Ziv RD function \cite{wyner1976}: Letting $W_j\in\mathcal{W}_j$ with $|\mathcal{W}_j|\le |\mathcal{S}_j|+1$ denote an auxiliary random variable that satisfies the Markov chain $W_j\markov S_j\markov S_{j'}$, we have
\[
R^{(j)}_{\text{WZ}}(D_j)=\min\limits_{P_{W_j|S_j}}\min_{\substack{g: \mathcal{S}_{j'}\times\mathcal{W}_{j}\rightarrow\hat{\mathcal{S}}_j\\ \mathbb{E}[d_j(S_j, g(S_{j'}, W_j)]\le D_j}} I(S_j; W_j|S_{j'}).
\]\smallskip
\item  Conditional RD function \cite{gray1972}: 
\[
R_{S_j|S_{j'}}(D_j)=\min_{\substack{P_{\hat{S}_j|S_1, S_2}\\ \mathbb{E}[d_j(S_j, \hat{S}_j)]\le D_j}} I(S_j; \hat{S}_j|S_{j'}).
\] 
\end{itemize}

\section{Hybrid Digital/Analog Coding for Correlated Sources over DM-TWCs}
Inspired by the coding method in \cite{kim2015}, here we propose a rate-one (i.e., $n=k$) two-way joint source-channel hybrid coding scheme as depicted in Fig.~\ref{fig:Hybridblcok}. 
We first map the source block $S_j^n$ to a digital codeword $U_j^n(M_j)$ with index $M_j$ for $j=1, 2$. 
The channel input $X_j^n$ is then generated via a symbol-by-symbol map $f_j$ which combines the digital information $U_j^n(M_1)$ and the raw (or analog) information $S_j^n$.\footnote{In this scheme, the channel inputs do not make use of past received channel outputs (i.e., no interactive adaptive coding is used).}
Upon receiving $Y_j^n$, the channel decoder~$j$ estimates the codeword index $M_{j'}$ based on all available information, where $j, j'=1, 2$ with $j'\neq j$. 
The decoded codeword $U_{j'}(\hat{M}_{j'})$ and source $S_j^n$ are together passed through a symbol-by-symbol map $g_j$ to produce source reconstruction $\hat{S}_{j'}^n$. 
This hybrid coding scheme allows us to make use of the source correlation in generating the channel inputs which facilitates user cooperation and hence improves coding performance. 
Also, prior achievability results can be realized by the proposed hybrid coding system. 

\begin{figure}[!t]
\centering
\includegraphics[draft=false, scale=0.47]{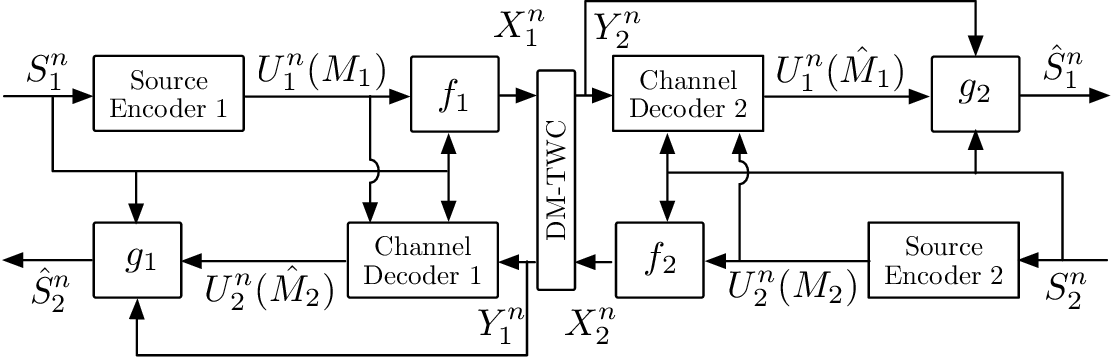}
\caption{Rate-one hybrid coding scheme for the transmission of correlated sources over DM-TWCs.}
\label{fig:Hybridblcok}
\end{figure}

\subsection{Rate-One Hybrid Coding Scheme}
\begin{theorem}
A distortion pair $(D_1, D_2)$ is in the achievable distortion region for the rate-one transmission of the source pair $(S_1, S_2)$ over a DM-TWC with transition probability $P_{Y_1, Y_2|X_1, X_2}$ if
\begin{subequations}
\label{eq:hybrid}
\begin{IEEEeqnarray}{rCl}
I(S_1; U_1|S_2, U_2) < I(U_1; Y_2|S_2, U_2)\label{eq:hybrida}\\
I(S_2; U_2|S_1, U_1) < I(U_2; Y_1|S_1, U_1)\label{eq:hybridb}
\end{IEEEeqnarray}
\end{subequations}
{for some joint distribution on the random variables $S_1$, $S_2$, $U_1$, $U_2$, $X_1$, $X_2$, $Y_1$, and $Y_2$ given by $P_{S_1, S_2}\allowbreak P_{U_1|S_1}\allowbreak P_{U_2|S_2}\allowbreak P_{X_1|U_1, S_1}\allowbreak P_{X_2|U_2, S_2}\allowbreak P_{Y_1, Y_2|X_1, X_2}$, two encoding functions $X_j=f_j(U_j, S_j)$, and two decoding functions $\hat{S}_{j'}=g_j(U_{j'}, U_j, S_j, Y_j)$, where $j, j'=1, 2$ with $j\neq j'$, such that $\mathbb{E}[d_j(S_j, \hat{S}_j)]\le D_j$.}
\end{theorem}

\begin{proof}[Sketch of Proof]
The proof is an extension of the one-way hybrid coding scheme of \cite{kim2015}. 
The detailed analysis is omitted here due to space limitations. 

Let $\mathcal{T}_{\epsilon}^{(n)}$ denote the set of jointly typical sequences with parameters $n$ and $\epsilon$ as defined in \cite{kim2011}; the domain of the sequences in $\mathcal{T}_{\epsilon}^{(n)}$ should be clear from the context and hence is omitted for the sake of brevity. For $j=1, 2$, we define $2^{\mli{nR}_j}$ as the size of user~$j$'s codebook.

\textit{\underline{Codebook Generation}:} Let $\epsilon>\epsilon^{'}>0$. For $j, j'=1, 2$ with $j\neq j'$, we fix $P_{U_j|S_j}$, encoding function $f_j(u_j, s_j)$, and decoding function $g_{j'}(u_j, u_{j'}, s_{j'}, y_{j'})$ such that $\mathbb{E}[d_j(S_j, \hat{S}_j)]\le D_j$. 
For $m_j=1, 2, \cdots, 2^{\mli{nR}_j}$, length-$n$ codewords $u_j^n(m_j)$'s are randomly drawn according to the distribution $\prod_{i=1}^n P_{U_j}(u_{j, i})$.
The codebooks $C_1$ and $C_2$, where $C_j=\{u_j^n(m_j): m_j=1, 2, \dots, 2^{\mli{nR}_j}\}$, are revealed to both users.

\textit{\underline{Encoding}:} 
For $j=1, 2$, upon observing $s_j^n$, the encoder of user~$j$ finds an index $m_j$ such that $s_j^n$ and $u_j^n(m_j)$ are jointly typical, i.e., $(s_j^n, u_j^n(m_j))\in\mathcal{T}_{\epsilon^{'}}^{(n)}$. If there is more than one such index, the encoder chooses one of them at random. 
If there is no such index, it chooses an index at random from $\{1, 2, \dots, 2^{\mli{nR}_j}\}$. The transmitter of user~$j$ then sends $x_{j, i}=f_{j, i}(u_{j, i}(m_j), s_{j, i})$, $i=1, 2, \dots, n$, over the DM-TWC. 

\textit{\underline{Decoding}:} 
For $j, j'=1, 2$ with $j\neq j'$, upon receiving $y_j^n$, the decoder of user~$j$ looks for a unique index $\hat{m}_{j'}$ such that $(u_{j'}^n(\hat{m}_{j'}), u_j^n(m_j), s_j^n, y_j^n)\in\mathcal{T}_{\epsilon}^{(n)}$. 
If there is more than one choice, the decoder chooses one of them at random. 
If there is no such unique choice, it chooses one at random from $\{1, 2, \dots, 2^{\mli{nR_j}}\}$. 
The reconstruction of $s_{j'}^n$ is then produced symbol-by-symbol via $\hat{s}_{j', i}=g_{j}(u_{j', i}(\hat{m}_{j'}), u_{j, i}(m_j), s_{j, i},  y_{j, i})$ for $i=1, 2, \dots, n$.   

\textit{\underline{Analysis of Expected Distortion}:} 
Let $M_j$ and $\hat{M}_j$ denote the (random) index chosen by encoder~$j$ and the corresponding decoded index, respectively. 
Define the error event $\mathcal{E}\triangleq \mathcal{E}_1 \cup \mathcal{E}_2$, where
$\mathcal{E}_1=\{(S_1^n, S_2^n, U_1^n(\hat{M}_1), U_2^n(M_2), Y_2^n)\notin\mathcal{T}_{\epsilon}^{(n)}\}$
and
$
\mathcal{E}_2=\{(S_1^n, S_2^n, U_1^n(M_1), U_2^n(\hat{M}_2), Y_1^n)\notin\mathcal{T}_{\epsilon}^{(n)}\}.
$
The expected distortion for users~$j$'s source reconstruction can be bounded as
\begin{IEEEeqnarray}{l}
\mathbb{E}[d_j(S_j^n, \hat{S}_j^n)]\nonumber\\
\ \ =\Pr(\mathcal{E})\cdot\mathbb{E}[d_j(S_j^n, \hat{S}_j^n)|\mathcal{E}]+\Pr(\bar{\mathcal{E}})\cdot\mathbb{E}[d_j(S_j^n, \hat{S}_j^n)|\bar{\mathcal{E}}]\nonumber\\
\ \ \le \Pr(\mathcal{E})\cdot d_{\max, j}+\Pr(\bar{\mathcal{E}})\cdot\mathbb{E}[d_j(S_j^n, \hat{S}_j^n)|\bar{\mathcal{E}}], \nonumber
\end{IEEEeqnarray}
where $\bar{\mathcal{E}}$ denotes the complement of $\mathcal{E}$, $d_{\max, j}\triangleq\max_{s\in\mathcal{S}_j, \hat{s}\in\hat{\mathcal{S}}_j}d_j(s, \hat{s})$, and $\mathbb{E}[d_j(S_j^n, \hat{S}_j^n)|\bar{\mathcal{E}}]$ is upper bounded by $(1+\epsilon)\cdot \mathbb{E}[d_j(S_j, \hat{S}_j)]$ due to the typical average lemma \cite[Sec. 2.4]{kim2011}.
Together with the assumption that $\mathbb{E}[d_j(S_j, \hat{S}_j)]\le D_j$, we obtain that $\mathbb{E}[d_j(S_j^n, \hat{S}_j^n)]\le \Pr(\mathcal{E})\cdot d_{\max, j}+\Pr(\bar{\mathcal{E}})\cdot (1+\epsilon)D_j$. 
It then suffices to show that $\Pr(\mathcal{E}_1)\to 0$ and $\Pr(\mathcal{E}_2)\to 0$ as $n\to\infty$ since $\Pr(\mathcal{E})\le \Pr(\mathcal{E}_1)+\Pr(\mathcal{E}_2)$.
We first analyze $\Pr(\mathcal{E}_1)$ by considering the following four events: 
\begin{IEEEeqnarray}{rCl}
\mathcal{F}_{1}&=&\{(S_1^n, U_1^n(m_1))\notin\mathcal{T}_{\epsilon^{'}}^{(n)}\ \text{for all}\ m_1\},\nonumber\\
\mathcal{F}_{2}&=&\{(S_2^n, U_2^n(m_2))\notin\mathcal{T}_{\epsilon^{'}}^{(n)}\ \text{for all}\ m_2\},\nonumber\\
\mathcal{F}_{3}&=&\{(S_1^n, S_2^n, U_1^n(M_1), U_2^n(M_2), Y_2^n)\notin\mathcal{T}_{\epsilon}^{(n)}\},\nonumber\\
\mathcal{F}_{4}&=&\{\exists\ m_1\neq M_1\ \text{s.t.}\ (U_1^n(m_1), U_2^n(M_2), S_2^n, Y_2^n)\in\mathcal{T}_{\epsilon}^{(n)}\}.\nonumber
\end{IEEEeqnarray}
Here, since the joint typicality in $\overline{\mathcal{F}_3}$ implies that $(U_1^n(M_1), U_2^n(M_2), S_2^n, Y_2^n)\in\mathcal{T}_{\epsilon}^{(n)}$, the event $\overline{\mathcal{F}_3}\cap\overline{\mathcal{F}_4}$ implies that $\hat{M}_1=M_1$. In other words, we obtain that $(S_1^n, S_2^n, U_1^n(\hat{M}_1), U_2^n(M_2), Y_2^n)\in\mathcal{T}_{\epsilon}^{(n)}$ under $\overline{\mathcal{F}_3}\cap\overline{\mathcal{F}_4}$.  
This shows that $\overline{\mathcal{F}_3}\cap\overline{\mathcal{F}_4}\subseteq\overline{\mathcal{E}_1}$ and hence $\mathcal{E}_1\subseteq\mathcal{F}_3\cup\mathcal{F}_4$, from which we obtain $\mathcal{E}_1\subset \mathcal{F}_1\cup \mathcal{F}_2\cup (\overline{\mathcal{F}_1\cup \mathcal{F}_2}\cap\mathcal{F}_3)\cup\mathcal{F}_4$.

Then by the union bound, we have 
\begin{IEEEeqnarray}{l}
\Pr(\mathcal{E}_1)\nonumber\\
\ \ \le \Pr(\mathcal{F}_1)+\Pr(\mathcal{F}_2)+\Pr(\overline{\mathcal{F}_1\cup \mathcal{F}_2}\cap\mathcal{F}_3)+\Pr(\mathcal{E}_4). \label{eq:ub} 
\IEEEeqnarraynumspace
\end{IEEEeqnarray}
The main part of the proof is to show that under the conditions that $R_1> I(S_1; U_1)$, $R_2> I(S_2; U_2)$, and $R_1<I(U_1; U_2, S_2, Y_2)$, each term on the right-hand-side of \eqref{eq:ub} tends to zero as $n\to\infty$, implying that $\Pr(\mathcal{E}_1)\rightarrow 0 $ as $n\to\infty$.
Combining the two inequalities $R_1> I(S_1; U_1)$ and $R_1<I(U_1; U_2, S_2, Y_2)$ then results in \eqref{eq:hybrida}. 
Due to symmetry, the same argument can be applied to $\Pr(\mathcal{E}_2)$ to obtain \eqref{eq:hybridb}.  
Thus, if \eqref{eq:hybrid} is satisfied, then $((1+\epsilon)D_1, (1+\epsilon)D_2)$ is achievable for all $\epsilon>0$, so $(D_1, D_2)$ is in the achievable distortion region.
\end{proof}

\subsection{Special Cases}
The proposed joint source-channel coding scheme includes as special cases uncoded transmission, rate-one separate source-channel coding, and correlation-preserving coding for (almost) lossless transmission.

\begin{itemize}
\item[(i)] \textbf{Uncoded transmission scheme}: Setting $U_j=\text{constant}$, $f_j(u_{j, i}, s_{j, i})=s_{j, i}$, and $g_{j'}(u_{j, i}, u_{j', i}, s_{j', i}, y_{j', i})=\tilde{g}_{j'}(s_{j', i}, y_{j', i})=\hat{s}_{j, i}$ for some decoding function $\tilde{g}_{j'}$ (e.g., maximum-a-posteriori (MAP) or minimum mean-square-error (MMSE) decoder) such that $\mathbb{E}[d_{j}(S_{j}, \hat{S}_{j})]\le D_{j}$, we obtain the uncoded scheme via the hybrid coding architecture. Note that, under this setup, the inequalities in (1) and (2) trivially hold since the left-hand-sides of both inequalities are zero.\smallskip

\item[(ii)] \textbf{SSCC for the lossy transmission of independent sources}: 
Define $\mathcal{V}_j\triangleq\mathcal{X}_j$ and $\mathcal{W}_j\triangleq\hat{\mathcal{S}}_j$, and let the joint distribution of $(S_1, S_2, U_1, U_2, X_1, X_2)$ be given by
\begin{IEEEeqnarray}{l}
P_{S_1, S_2, U_1, U_2, X_1, X_2}=P_{S_1}\cdot P_{S_2}\nonumber\\
\ \ \ \ \cdot\underbrace{P_{V_1}\cdot P_{W_1|S_1}}_{\triangleq P_{U_1|S_1}}\cdot\underbrace{P_{V_2}\cdot P_{W_2|S_2}}_{\triangleq P_{U_2|S_2}}\cdot\underbrace{P_{X_1|V_1}}_{\triangleq P_{X_1|S_1, U_1}}\cdot\underbrace{P_{X_2|V_2},}_{\triangleq P_{X_2|S_2, U_2}}\nonumber\IEEEeqnarraynumspace 
\end{IEEEeqnarray}
where $U_j\triangleq(V_j, W_j)$, $P_{V_j}$ is arbitrary, $P_{W_j|S_j}$ achieves the standard RD function at distortion level $D_j$ for $S_j$, and $P_{X_j|V_j}(x_j|v_j)=\mathbf{1}(x_j=v_j)$ (here $\mathbf{1}(\cdot)$ denotes the indicator function). 
In this setup, the encoding function is $x_{j, i}=f_j(u_{j, i}, s_{j, i})=v_{j, i}$ since $X_j=V_j$. 
Furthermore, the decoding function $g_{j'}$ is chosen as $\hat{s}_{j, i}=g_{j'}(u_{j, i}, u_{j', i}, s_{j', i}, y_{j', i})=w_{j, i}$, i.e., $\hat{S}_{j}=W_{j}$, thus satisfying $\mathbb{E}[d_{j}(S_{j}, \hat{S}_{j})]\le D_{j}$.  
The sufficient conditions in \eqref{eq:hybrid} become $R^{(1)}(D_1)=I(S_1; \hat{S}_1)<I(X_1; Y_2|X_2)$ and $R^{(2)}(D_2)=I(S_2; \hat{S}_2)<I(X_2; Y_1|X_1)$.
An achievability result for independent sources based on separate source-channel coding is thus obtained. \smallskip

\item[(iii)] \textbf{SSCC for the lossy transmission of correlated sources} \cite{jjw2017}: 
Define $\mathcal{V}_j\triangleq\mathcal{X}_j$, and let the joint distribution of $(S_1, S_2, U_1, U_2, X_1, X_2)$ be given by
\begin{IEEEeqnarray}{l}
P_{S_1, S_2, U_1, U_2, X_1, X_2}=P_{S_1, S_2}\nonumber\\
\ \ \ \ \cdot\underbrace{P_{V_1}\cdot P_{W_1|S_1}}_{\triangleq P_{U_1|S_1}}\cdot\underbrace{P_{V_2}\cdot P_{W_2|S_2}}_{\triangleq P_{U_2|S_2}}\cdot\underbrace{P_{X_1|V_1}}_{\triangleq P_{X_1|S_1, U_1}}\cdot\underbrace{P_{X_2|V_2},}_{\triangleq P_{X_2|S_2, U_2}} \nonumber\IEEEeqnarraynumspace
\end{IEEEeqnarray}
where $U_j\triangleq (V_j, W_j)$, $P_{V_j}$ is arbitrary, $P_{W_j|S_j}$ achieves the Wyner-Ziv RD function at distortion level $D_j$ and has associated decoding function $\tilde{g}_{j'}(w_j, s_{j'})$ for $S_j$, and $P_{X_j|V_j}(x_j|v_j)=\mathbf{1}(x_j=v_j)$. 
We also set $x_{j, i}=f_j(u_{j, i}, s_{j, i})=v_{j, i}$ since $X_j=V_j$. 
The decoding function is set as $\hat{s}_{j, i}=g_{j'}(u_{j, i},  u_{j', i}, s_{j', i}, y_{j', i})=\tilde{g}_{j'}(w_{j, i}, s_{j', i})$, where $\tilde{g}_{j'}(\cdot, \cdot)$ is the decoding function that achieves the Wyner-Ziv RD function.  
Note that the Markov chain relationship $W_{j'}\markov S_{j'}\markov S_j$ holds.
Then, the sufficient conditions in \eqref{eq:hybrid} become $R_{\text{WZ}}^{(1)}(D_1)<I(X_1; Y_2|X_2)$ and $R_{\text{WZ}}^{(2)}(D_2)<I(X_2; Y_1|X_1)$. 
The achievability result of \cite{jjw2017} for correlated sources based on separate source-channel coding is hence recovered.\smallskip

\item[(iv)] \textbf{Correlation-preserving coding scheme for (almost) lossless transmission of correlated sources \cite{gunduz2009}}: 
Suppose that $\mathcal{S}_j=\hat{\mathcal{S}}_j$ and consider the Hamming distortion measure \cite[Sec. 3.6]{kim2011}. 
Let $\mathcal{W}_j\triangleq\hat{\mathcal{S}}_j$ and $\mathcal{V}_j\triangleq\mathcal{X}_j$. 
Set the joint distribution of $(S_1, S_2, U_1, U_2, X_1, X_2)$ as
\begin{IEEEeqnarray}{l}
P_{S_1, S_2, U_1, U_2, X_1, X_2}=P_{S_1, S_2}\nonumber\\
\ \underbrace{P_{V_1|S_1}\cdot P_{W_1|S_1}}_{\triangleq P_{U_1|S_1}}\cdot\underbrace{P_{V_2|S_2}\cdot P_{W_2|S_2}}_{\triangleq P_{U_2|S_2}}\cdot\underbrace{P_{X_1|V_1}}_{\triangleq P_{X_1|S_1, U_1}}\cdot\underbrace{P_{X_2|V_2},}_{\triangleq P_{X_2|S_2, U_2}}\nonumber\IEEEeqnarraynumspace
\end{IEEEeqnarray}
where {  $U_j\triangleq(V_j, W_j)$, $P_{V_j|S_j}$ is arbitrary, $P_{W_j|S_j}(w_j|s_j)=\mathbf{1}(s_j=w_j)$, and $P_{X_j|V_j}(x_j|v_j)=\mathbf{1}(x_j=v_j)$. 
Clearly, we have $x_{j, i}=f_j(u_{j, i}, s_{j, i})=v_{j, i}$ since $X_j=V_j$. 
Furthermore, the decoding function is chosen as $\hat{s}_{j, i}=g_{j'}(u_{j, i}, u_{j', i}, s_{j', i}, y_{j', i})=w_{j, i}$ and hence we have $\hat{S}_{j}=W_{j}=S_{j}$. 
In this setup, $\Pr\big(\mathcal{E}\big)\to 0$ (as $n\to\infty)$ implies that $\Pr\big((S_j^n, W_j^n)\in\mathcal{T}_{\epsilon}^{(n)}\big)\to 1$ for $j=1, 2$. 
Together with the fact that $W_j=\hat{S}_j$ and the choice of $P_{W_j|S_j}$, we obtain $\Pr\big(\{S_1^n\neq\hat{S}_1^n\}\cup\{S_2^n\neq\hat{S}_2^n\} \big)\to 0$ as $n\to\infty$.  
The sufficient conditions in \eqref{eq:hybrid} become $H(S_1|S_2)<I(X_1; Y_2|X_2, S_2)$ and $H(S_2|S_1)<I(X_2; Y_1|X_1, S_1)$. 
The achievability result in \cite[Cor.~8.1]{gunduz2009} (without coded time-sharing) is hence recovered.}
\end{itemize}

{The next simple example shows that the proposed scheme strictly generalizes the above coding methods for correlated sources.}
\begin{example}
Assuming binary alphabet for all variables, let sources $S_1$ and $S_2$ have joint distribution $P_{S_1, S_2}(0, 0)=P_{S_1, S_2}(0, 1)=P_{S_1, S_2}(1, 1)=1/3$ and $P_{S_1, S_2}(1, 0)=0$ and be sent over the two-way channel described by $Y_1=X_1\oplus X_2\oplus Z$ and $Y_2=X_1\cdot X_2$, where $\oplus$ denotes binary addition and $Z$ is binary noise with $P_{Z}(1)=0.05$ that is independent of $S_j$'s and $X_j$'s. 
In the proposed hybrid coding scheme, we use uncoded transmission for the direction from user~1 to~2 and standard SSCC for the reverse direction. 
This configuration can be shown to achieve distortion pair $(D_1, D_2)=(0, 0)$ under the Hamming distortion measure. 
However, no SSCC scheme can attain this distortion pair since $H(S_1|S_2)<I(X_1; Y_2|X_2)$ and $H(S_2|S_1)<I(X_2; Y_1|X_1)$ cannot hold simultaneously; similarly the correlation-preserving coding scheme cannot achieve (almost) lossless reconstruction. 
Moreover, using uncoded transmission in both directions yields the distortion pair $(D_1, D_2)=(0, 0.033)$. 
\end{example}

\begin{remark}
The achievable distortion region of the proposed rate-one hybrid coding scheme may be further enlarged by taking into account a time-sharing parameter.
Also, Theorem~1 can be extended to the case of bandwidth mismatch, i.e., $n\neq k$, by enlarging the source and channel symbol alphabets (with size $|S_j|^k$ and $|X_j|^n$ for rate-$k/n$ transmission). 
\end{remark}

\section{A Converse Result and A Joint Source-Channel Coding Theorem}


\begin{theorem}
\label{chIIlem:gaouterbound}
If a rate-$k/n$ JSSC achieves the distortion levels $D_1$ and $D_2$ for the lossy transmission of correlated sources $S_1$ and $S_2$ over a DM-TWC with transition probability $P_{Y_1, Y_2|X_1, X_2}$, then we have
\begin{subequations}
\begin{IEEEeqnarray}{rCl}
k\cdot R_{S_1|S_2}(D_1)\le  n\cdot I(X_1; Y_2|X_2),\label{eqa}\\ 
k\cdot R_{S_2|S_1}(D_2)\le  n\cdot I(X_2; Y_1|X_1),\label{eqb}
\end{IEEEeqnarray}
\end{subequations}
for some joint distribution on the random variables $S_1$, $S_2$, $X_1$, $X_2$, $Y_1$, and $Y_2$ given by $P_{S_1, S_2}P_{X_1, X_2}P_{Y_1, Y_2|X_1, X_2}$. 

\end{theorem}
\begin{proof}
Consider a genie-aided system for the transmission from user~1 to~2, where user~$1$ additionally has access to the decoder side-information $S_{2}^k$.
The encoder of the genie-aided system is of the form $f_{1, i}: \mathcal{S}_1^k\times \mathcal{S}_2^k\times\mathcal{Y}_1^{i-1}$ for $i=1, 2, \dots, n$.
Clearly, the original system in Definition~1 is a special case of the genie-aided system and thus the genie-aided system can be used to lower bound the distortion of the original system.

Assume that the genie-aided system has expected distortion $\mathbb{E}[d_1(S_1^k, \hat{S}_1^k)]\le D_1$.  
Then, we have
{\allowdisplaybreaks
\begin{IEEEeqnarray}{l}
k\cdot R_{S_1|S_2}(D_1)\nonumber\\
\ \ \le k\cdot R_{S_1|S_2}\left(k^{-1}\sum\limits_{m=1}^k \mathbb{E}\left[d_1(S_{1, m}, \hat{S}_{1, m})\right]\right)\label{cc-8}\\
\ \ \le \sum\limits_{m=1}^k R_{S_1|S_2}\left(\mathbb{E}[d_1(S_{1, m}, \hat{S}_{1, m})]\right)\label{cc-1}\\
\ \ \le \sum\limits_{m=1}^k I(S_{1, m}; \hat{S}_{1, m}|S_{2, m})\label{cc-2}\\
\ \ \le \sum\limits_{m=1}^k I(S_{1, m}; Y_2^n, S_2^k|S_{2, m})\label{cc-3}\\
\ \ = \sum\limits_{m=1}^k H(S_{1, m}|S_{2, m})-H(S_{1, m}|Y_2^n, S_{2}^k)\nonumber\\
\ \ \le \sum\limits_{m=1}^k H(S_{1, m}|S_2^k, S_1^{m-1})-H(S_{1, m}|Y_2^n, S_{2}^k, S_1^{m-1})\label{cc-4}\IEEEeqnarraynumspace\\
\ \ = I(S_1^k; Y_2^n|S_2^k)\nonumber\\
\ \ = \sum_{i=1}^n I(S_1^k; Y_{2, i}|S_2^k, Y_2^{i-1})\nonumber\\
\ \ \le \sum_{i=1}^n H(Y_{2, i}|X_{2, i})-H(Y_{2, i}|S_2^k, Y_2^{i-1}, S_1^k, X_{2, i})\label{cc-5}\\
\ \ \le \sum_{i=1}^n H(Y_{2, i}|X_{2, i})-H(Y_{2, i}|X_{1, i}, X_{2, i})\label{cc-7}\\
\ \ = \sum\limits_{i=1}^n I(X_{1, i}; Y_{2, i}|X_{2, i})\nonumber\\
\ \ \le n\cdot I(X_1; Y_2|X_2)\label{cc-6},
\end{IEEEeqnarray}
where \eqref{cc-8} follows since $R_{S_1|S_2}(D_1)$ is non-increasing and the assumption on the expected distortion, \eqref{cc-1} and \eqref{cc-2} are respectively due to convexity and the definition of conditional RD function, \eqref{cc-3} follows from the data-processing inequality, \eqref{cc-4} holds by the Markov chain relationship $S_{1, m}\markov S_{2, m}\markov (S_2^k, S_1^{m-1})$, \eqref{cc-5} holds since conditioning reduces entropy and $X_{2, i}$ is a function of $(Y_2^{i-1}, S_2^k)$, \eqref{cc-7} holds since the channel is memoryless, \eqref{cc-6} holds with $P_{X_1, X_2}=n^{-1}\sum_{i=1}^n P(X_{1, i}, X_{2, i})$  since $I(X_{1, i}; Y_{2, i}|X_{2, i})$ is concave in $P_{X_{1, i}, X_{2, i}}$.  
}

Due to symmetry, we obtain the condition in \eqref{eqb} by swapping the roles of users~1 and~2. 
Combining the bounds proves the theorem.  
\end{proof}


\begin{remark}
Consider $\hat{\mathcal{S}_j}=\mathcal{S}_j$. Under the Hamming distortion measure, Theorem~2 reduces to \cite[Proposition 8.2]{gunduz2009} for (almost) lossless transmission since $\Pr\big(S_j^k\neq\hat{S}_j^k\big)=0$ is a stricter requirement than $k^{-1}\sum_{m=1}^k\Pr\big(S_{j, m}\neq \hat{S}_{j, m}\big)=0$ and $R_{S_j|S_{j'}}(0)=H(S_j|S_{j'})$ for $j, j'=1, 2$ with $j\neq j'$. 
\end{remark}

Let $\tilde{R}_j$ denote the channel coding rate of user~$j$ for $j=1, 2$. 
In \cite{shannon1961}, Shannon derived capacity inner and outer bounds for DM-TWCs:
\begin{subequations}
\vspace{-0.3cm}
\begin{IEEEeqnarray}{rCl}
\tilde{R}_1&\le & I(X_1; Y_2|X_2),\nonumber\\
\tilde{R}_2&\le & I(X_2; Y_1|X_1),\nonumber
\end{IEEEeqnarray}
\end{subequations}
\vspace{-0.1cm}where $X_1$ and $X_2$ are independent inputs in the inner bound, while in the outer bound $X_1$ and $X_2$ are arbitrarily correlated.
Although the two capacity bounds in general do not coincide, the capacity region can be exactly determined for channels with symmetry properties \cite{han1984, cheng2014, song2016, chaaban2017, jjw2018}. 
In \cite{jjw2017}, a complete joint source-channel coding theorem was derived for independent sources and DM-TWCs for which interactive adaptive coding cannot enlarge the capacity region, i.e., Shannon's capacity inner bound is tight. 
Based on the achievability result in \cite[Lemma~1]{jjw2017} (special case (iii) of Section III-B) and Theorem~\ref{chIIlem:gaouterbound} above, we obtain another joint source-channel coding theorem as follows. 

\begin{theorem}[\textbf{A Joint Source-Channel Coding Theorem}]
\label{JSCC} 
For the rate-$k/n$ transmission of correlated source pair $(S_1, S_2)$ whose Wyner-Ziv and conditional RD functions are equal, i.e., $R^{(j)}_{\emph{WZ}}(D_j)=R_{S_j|S_{j'}}(D_j)$ for $j, j=1, 2$ with $j\neq j'$, over a DM-TWC with transition probability $P_{Y_1, Y_2|X_1, X_2}$ for which interactive adaptive coding cannot enlarge the capacity region, a distortion pair $(D_1, D_2)$ is achievable if and only if  
\begin{IEEEeqnarray}{rCl}
k\cdot R^{(1)}_{\emph{WZ}}(D_1)&\le & n \cdot I(X_1; Y_2|X_2),\nonumber\\
k\cdot R^{(2)}_{\emph{WZ}}(D_2)&\le & n \cdot I(X_2; Y_1|X_1),\nonumber
\end{IEEEeqnarray}
for some joint distribution on the random variables $S_1$, $S_2$, $X_1$, $X_2$, $Y_1$, and $Y_2$ given by $P_{S_1, S_2}P_{X_1}P_{X_2}P_{Y_1, Y_2|X_1, X_2}$. 
\label{JSC}
\end{theorem}

We end this section with two examples that satisfy Theorem~3.
In the first example, the alphabets are not finite but one can employ the discretization procedures in \cite[Sec. 3.8 and Sec. 3.4]{kim2011} to obtain their quantized descriptions. Our results can be directly applied to this discretized Gaussian system.

\begin{example}
Transmitting correlated Gaussian sources over memoryless TWCs with additive white Gaussian noise (AWGN) \cite{han1984} under the squared-error distortion measure. 
\end{example}

\begin{example}
Transmitting binary correlated sources with $Z$-channel correlation \cite{nikos2014} over binary additive noise DM-TWCs \cite{song2016} under the Hamming distortion measure. 
\end{example}

\section{Conclusions}
We extended the hybrid coding scheme of \cite{kim2015} for one-way channels to a TWC setup such that channel state and decoder side-information are taken into accounts.
The proposed scheme not only includes prior coding techniques as special cases but also leverages source correlation against channel noise. 
When the channel capacity cannot be enlarged by adaptive coding, we further showed that the SSCC scheme given in \cite{jjw2017} is optimal for coding correlated sources for which the Wyner-Ziv and conditional rate-distortion functions coincide. 
Future work includes identifying scenarios where partial correlation-preserving schemes on channel inputs can improve source reconstruction, introducing interactive adaptive coding (in terms of previously received channel outputs) in the hybrid coding system, and establishing posterior-matching-like coding schemes \cite{shayevitz2011} for transmitting correlated sources over general DM-TWCs.   

\bibliographystyle{IEEEtran}
\bibliography{thesis}

\end{document}